\documentclass[twoside,leqno,11pt]{article}
\usepackage[letterpaper, margin = 2cm]{geometry}

\usepackage{amsthm}
\usepackage{amsmath,amssymb}
\usepackage{hyperref}
\usepackage[capitalise]{cleveref}
\crefname{figure}{Figure}{Figures}
\usepackage{thm-restate}
\usepackage{enumitem}
\usepackage{comment}
\usepackage{algorithm}

\usepackage{algpseudocodex}
\algrenewcommand\alglinenumber[1]{\normalfont\tiny\sffamily #1}
\algrenewcommand\algorithmicfunction{\textbf{subroutine}}

\newtheorem{definition}{Definition}[section]

\usepackage{lineno}

\usepackage{graphicx}
\graphicspath{{./figures}}

\usepackage{xspace}

\newcommand{\neghphantom}[1]{\settowidth{\dimen0}{#1}\hspace*{-\dimen0}}

\newcommand{\mypar}[1]{\medskip\noindent{\bf #1}~}

\newcommand{\etal}{\textit{et~al.\@}\xspace}

\newcommand{\f}{Fr\'echet\xspace}
\newcommand{\dF}{\ensuremath{d_\mathrm{F}}}
\newcommand{\eps}{\varepsilon\xspace}
\newcommand{\R}{\ensuremath{\mathbb{R}}}

\newcommand{\FS}[1][\Delta]{\ensuremath{#1}\textnormal{-FS}}

\newcommand{\Yes}{\text{``Yes''}\xspace}
\newcommand{\No}{\text{``No''}\xspace}
\newcommand{\Timeout}{\text{``Timeout''}\xspace}

\usepackage{todonotes}

\newtheorem{thm}{Theorem}

\let\originalleft\left
\let\originalright\right
\renewcommand{\left}{\mathopen{}\mathclose\bgroup\originalleft}
\renewcommand{\right}{\aftergroup\egroup\originalright}

\makeatletter
\newcommand{\IfRestatedTF}[2]{\ifthmt@thisistheone #2\else #1\fi}
\makeatother

\hypersetup{
  colorlinks=true,
  linkcolor=red!70!black,
  linktoc=all,
  citecolor=blue,
  pdfpagemode=FullScreen,
}

\makeatletter
\renewcommand\thefootnote{\textsuperscript{\@fnsymbol\c@footnote}}
\let\old@thanks\thanks 
\DeclareRobustCommand\thanks[2][]{
  \AddToHook{begindocument/end}{
    \if\relax#1\relax%
      \footnotemark%
    \else%
      \protect\refstepcounter{footnote}\protect\label{#1}%
    \fi%
    \protected@xdef\@thanks{%
      \@thanks\protect\footnotetext[\the\c@footnote]{#2}%
    }%
  }%
}
\AddToHook{begindocument/begin}{\let\thanks\old@thanks} 
\let\old@maketitle\maketitle
\def\maketitle{\old@maketitle\def\thefootnote{\@arabic\c@footnote}}
\def\footnoterule{\kern-.4\p@\hrule\@width 5pc\kern3\p@\kern-\footnotesep}
\makeatother

\newtheorem{lemma}{Lemma}

\title{Computing the Fréchet Distance When Just One Curve is $c$-Packed:\\A Simple Almost-Tight Algorithm}
\author{%
    Jacobus Conradi\ref{affil:bonn} \and %
    Ivor van der Hoog\ref{affil:itu} \and %
    Thijs van der Horst\ref{affil:uu}\ref{affil:tue} \and %
    Tim Ophelders\ref{affil:uu}\ref{affil:tue}
}

\begin{document}

\thanks[affil:bonn]{Institute of Computer Science, Universität Bonn, Germany}
\thanks[affil:itu]{Theoretical Computer Science, IT-University of Copenhagen, Denmark}
\thanks[affil:uu]{Department of Information and Computing Sciences, Utrecht University, the Netherlands}
\thanks[affil:tue]{Department of Mathematics and Computer Science, TU Eindhoven, the Netherlands}

\date{}
\maketitle

\begin{abstract}
\noindent We study approximating the continuous Fréchet distance of two curves with complexity $n$ and $m$, under the assumption that only one of the two curves is $c$-packed. 
Driemel, Har{-}Peled and Wenk DCG'12 studied Fréchet distance approximations under the assumption that both curves are $c$-packed. In $\mathbb{R}^d$, they prove a $(1+\varepsilon)$-approximation in $\tilde{O}(c\,\frac{n+m}{\varepsilon})$ time.
Bringmann and Künnemann IJCGA'17 improved this to $\tilde{O}(c\,\frac{n + m}{\sqrt{\varepsilon}})$ time, which they showed is near-tight under SETH.
Both algorithms have a hidden exponential dependency on the dimension $d$.
Recently, Gudmundsson, Mai, and Wong ISAAC'24 studied our setting where only one of the curves is $c$-packed.
They provide an involved $\tilde{O}((c+\varepsilon^{-1})(cn\varepsilon^{-2} + c^2m\varepsilon^{-7} + \varepsilon^{-2d-1}))$-time algorithm when the $c$-packed curve has $n$ vertices and the arbitrary curve has $m$.
In this paper, we show a simple technique to compute a $(1+\varepsilon)$-approximation in $\mathbb{R}^d$ in time $O(c\,\frac{n+m}{\varepsilon}\log\frac{n+m}{\varepsilon})$ when one of the curves is $c$-packed.
Our approach is not only simpler than previous work, but also significantly improves the dependencies on $c$, $\varepsilon$, and $d$ (which is only linear).
Moreover, it almost matches the asymptotically tight bound for when both curves are $c$-packed. 
Our algorithm is robust in the sense that it does not require knowledge of $c$, nor information about which of the two input curves is $c$-packed.
\end{abstract}

\paragraph{funding}
    Ivor van der Hoog is supported by the VILLUM Foundation grant (VIL37507) ``Efficient Recomputations for Changeful Problems''.
    Jacobus Conradi is funded by the iBehave Network: Sponsored by the Ministry of Culture and Science of the State of North Rhine-Westphalia and affiliated with the Lamarr Institute for Machine Learning and Artificial Intelligence.
    Tim Ophelders is supported by the Dutch Research Council (NWO) under project no.\ VI.Veni.212.260.%

\section{Introduction}

The Fréchet distance is a widely studied similarity measure for curves, with numerous real-world applications such as handwriting recognition \cite{sriraghavendra2007frechet}, map-matching~\cite{wenk06map_matching}, comparing coastlines~\cite{mascret2006coastline}, time series clustering~\cite{driemel16clustering}, or data analysis of outlines of shapes in geographic information systems~\cite{devogele2002new}, trajectories of moving objects~\cite{brakatsoulas2005map, buchin2020group,acmsurvey20, su2020survey}, air traffic~\cite{bombelli2017strategic} and protein structures~\cite{jiang2008protein}. 
Like the Hausdorff distance, the Fréchet distance is a bottleneck measure that outputs the distance between a pair of points from the two curves. However, unlike the Hausdorff distance, it respects the ordering of points along the curves, making it particularly well-suited for measuring similarity between moving data entries. 

Already in 1995, Alt and Godau~\cite{alt95continuous_frechet} presented a near-quadratic time algorithm for computing the Fréchet distance between two polygonal curves with $n$ and $m$ vertices, achieving a running time of $O(nm \log(n+m))$. Although there have since been incremental improvements~\cite{buchin17continuous_frechet,cheng25subquadratic_frechet}, strong evidence suggests that significantly faster algorithms are unlikely. In particular, Bringmann~\cite{bringmann14hardness} showed that, assuming the \emph{Strong Exponential Time Hypothesis}, no strongly subquadratic algorithm (i.e., with running time $O((nm)^{1-\delta})$ for any $\delta > 0$) exists. This lower bound also holds for algorithms that approximate the Fréchet distance within a factor less than~3, and even in one-dimensional settings~\cite{buchin19seth_says}. More recently, Cheng, Huang, and Zhang~\cite{cheng25subquadratic} gave the first (randomised) constant-factor approximation algorithm with subquadratic complexity, achieving a $(7 + \varepsilon)$-approximation in $\tilde{O}(nm^{0.99})$ time.

However, significantly faster algorithms are possible for restricted classes of curves. For example, when the curves consist of sufficiently long edges relative to their Fréchet distance, Gudmundsson, Mirzanezhad, Mohades, and Wenk~\cite{gudmundsson19long} gave a near-linear time exact algorithm. Additionally, several models have been proposed to capture the structure of ``realistic'' curves commonly found in applications. These include the notions of \emph{$\kappa$-bounded}~\cite{alt04planar_curves}, \emph{$\phi$-low density}~\cite{schwarzkopf96range_searching_low_density}, and \emph{$c$-packed} curves~\cite{driemel12realistic}.
Of these three realism assumptions, $c$-packedness is the most frequently studied~\cite{bringmann17cpacked,bruning2022faster, conradi20231+, driemel2023discrete, gudmundsson2023Cluster, gudmundsson2024Approximating, gudmundsson:mapmatching, gudmundsson23approximating_packedness, harpeled2025packeditreally, Hoog2024DataRealistic}.

In this work, we study approximating the \f distance for $c$-packed curves. A polygonal curve is said to be $c$-packed, for some value $c \in \mathbb{R}$, if the total length of the curve inside any disk of radius $r$ is at most $cr$. This notion, introduced by Driemel, Har-Peled, and Wenk~\cite{driemel12realistic}, captures the geometric structure of many real-world trajectories. They showed that for two $c$-packed curves with $n$ and $m$ vertices, a $(1 + \varepsilon)$-approximation of the Fréchet distance can be computed in $O{\left(c\, (n+m) \cdot \left(\frac{1}{\eps} + \log (n+m) \right)\right)}$ time.
The dependence on $\varepsilon$ was later improved by Bringmann and 
Künnemann~\cite{bringmann17cpacked}, who gave an algorithm running in $\tilde{O}(c\, (n+m)/\sqrt{\varepsilon})$ time, where the $\tilde{O}$ hides polylogarithmic factors in $n$ and $\varepsilon^{-1}$. Their result is conditionally optimal, assuming SETH, for curves in $\mathbb{R}^d$ when $d \geq 5$ and $m \in \Theta(n)$~\cite{bringmann14hardness}.
Gudmundsson, Sha, and Wong~\cite{gudmundsson23approximating_packedness} empirically studied the packedness of curves from real-world data sets, finding that many of them are $c$-packed for small values of $c$, often with $c \ll n$. This supports the practical relevance of algorithms whose running time depends on $c$ rather than on the input size directly.
The improved computational complexity for $c$-packed polygonal curves compared to arbitrary curves carries over to many problems related to the Fréchet distance. This ranges from the computation of the Fréchet distance of algebraic curves \cite{CONRADI2025Revisiting}, to the computation of partial similarity measures based on the Fréchet distance \cite{buchin2014shortcutNP, conradi2024kShortcut, Driemel2013Jaywalking} to the construction of approximate nearest neighbor data structures~\cite{conradi20231+}, minimum-center clustering \cite{bruning2022faster,conradi2025Subtrajectory, Hoog2025Clustering} and maximum-cardinality clustering \cite{gudmundsson2023Cluster}.

\mypar{When only one curve is $c$-packed.}
The algorithms of Driemel, Har-Peled, and Wenk~\cite{driemel12realistic} and Bringmann and Künnemann~\cite{bringmann17cpacked} assume that \emph{both} input curves are $c$-packed. More recently, two works have extended this setting to the asymmetric case, in which only one of the two curves is $c$-packed. 
Van der Hoog, Rotenberg, and Wong~\cite{Hoog2024DataRealistic} studied the simpler case of approximating the \emph{discrete} Fréchet distance under this assumption. They presented a $(1+\varepsilon)$-approximation algorithm for curves in $\mathbb{R}^d$, with a running time of $O{\left(c\,\frac{n+m}{\varepsilon} \log \frac{n}{\eps} \right)}$.
Gudmundsson, Mai, and Wong~\cite{gudmundsson2024Approximating} considered the \emph{continuous} Fréchet distance, which they showed to be considerably more complex. For a $c$-packed curve with $n$ vertices and an arbitrary curve with $m$ vertices in $\mathbb{R}^d$, they gave an algorithm with running time
$
\tilde{O}{\left(\left( c + \frac{1}{\varepsilon} \right) \left( \frac{cn}{\varepsilon^2} + \frac{c^2 m}{\varepsilon^7} + \frac{1}{\varepsilon^{2d+1}} \right) \right)}.
$
While their approach is the most general to date, its running time has significantly worse dependencies on $c$, $\varepsilon$, and the dimension $d$ compared to previous (less general) results.
To obtain this runtime, their algorithm relies on several sophisticated techniques and data structures, including layered graphs over the parameter space of $P$ and $Q$, as well as geometric range searching structures. 
These data structures form the bottleneck of their approach and are the source of their exponential dependency on the dimension $d$.

In this paper we give evidence that the asymptotic running time for the approximation of the Fréchet distance in this asymmetric setting should be the same as in the symmetric setting. 
We provide a strong structural theorem that bounds the complexity of the free-space between a $c$-packed and a non-$c$-packed curve. 
This strong characterisation of the free-space allows for a simple, straightforward and efficient approach for computing a $(1+\eps)$-approximation of the continuous Fréchet distance in $O(c \,\frac{n+m}{\eps} \log \frac{n+m}{\eps})$, if at least one of~$P$ or $Q$ is $c$-packed. This approach is not only more self-contained and simpler, it improves the dependency on $c$, $\eps$, and the dimension $d$ compared to~\cite{gudmundsson2024Approximating}.

\section{Preliminaries}

    A $d$-dimensional (polygonal) \emph{curve} $P$ defined by an ordered set $(p_1, \dots, p_n)\subset\R^d$ 
    is a piecewise-linear function $P \colon [1, n] \to \R^d$, where for $t\in[i,i+1]$ it is defined by $P(t)=p_i + (t-i)(p_{i+1}-p_i)$.   
    The points $p_1, \dots, p_n$ are the \emph{vertices} of $P$.
    The \emph{edges} of $P$ are the directed line segments from $p_i$ to $p_{i+1}$, for $i < n$.
    We denote by $P[x_1, x_2]$ the subcurve of $P$ over the domain $[x_1, x_2]$.
    A curve $P$ is said to be \emph{$c$-packed} if, for every $r > 0$, the total length of $P$ inside a ball of radius~$r$ is at most $cr$.

    A \emph{reparametrisation} of $[1, n]$ is a non-decreasing surjection $f \colon [0, 1] \to [1, n]$.
    Two reparametrisations $f$ and $g$ of $[1, n]$ and $[1, m]$, respectively, describe a \emph{matching} $(f, g)$ between two curves $P$ and $Q$ with $n$ and $m$ vertices, where for any $t \in [0, 1]$, the point $P(f(t))$ is matched to $Q(g(t))$.
    The (continuous) \emph{\f distance} $\dF(P, Q)$ between $P$ and $Q$ is defined as
    \[\dF(P,Q):=\min_{(f,g)}\max_{t\in[0,1]}\|P(f(t)) - Q(g(t))\|,\]
    where $(f,g)$ range over all matchings between $P$ and $Q$.
    
    For computing the \f distance one analyses the free space, a subset of the parameter space of the two given curves $P$ and $Q$.
    The \emph{parameter space} of curves $P$ and $Q$ with $n$ and $m$ vertices, respectively, is the rectangle $[1, n] \times [1, m]$ together with the regular grid whose cells are the squares $[i, i+1] \times [j, j+1]$ for integers $i$ and~$j$. 
    A point $(x, y)$ in the parameter space corresponds to the points $P(x)$ and $Q(y)$.
    Each cell $[i, i+1] \times [j, j+1]$ corresponds to the $i$-the edge of $P$ and the $j$-th edge of $Q$.

    \begin{definition}
          For any positive $\Delta$, we say a point $(x, y)$ in the parameter space of~$P$ and $Q$ is a $\Delta$-free point if $\| P(x) - Q(y) \| \leq \Delta$.     
    The \emph{$\Delta$-free space} ($\Delta$-FS) is the set of points $(x,y)\in[1,n]\times[1,m]$ such that $\|P(x)-Q(y)\|\leq\Delta$. 
     Finally, we say that a point $(x, y)$ is $\Delta$-\emph{reachable} if there is a monotone path (in $x$- and $y$-direction) from $(1,1)$ to $(x,y)$ which is contained in the $\Delta$-free space of $P$ and $Q$.
    \end{definition}

    Alt and Godau \cite{alt95continuous_frechet} observed that the $\Delta$-free space in every cell coincides with an ellipse intersected with that cell which can be computed in $O(1)$ time. The Fréchet distance between two curves $P$ and $Q$ is at most $\Delta$, if and only if $(n, m)$ is $\Delta$-reachable. This gives rise to the classical algorithm that decides whether $\dF(P,Q)\leq\Delta$ by discovering all $\Delta$-reachable points in the parameter space, one cell at a time.

\section{Contribution and Technical Overview}
We present a new algorithm that $(1+\varepsilon)$-approximates the Fréchet distance between two curves in $\mathbb{R}^d$.

\begin{restatable}{thm}{thmMain}\label{thm:main}
Let $P$ and $Q$ be polygonal curves in $\mathbb{R}^d$ with $n$ and $m$ vertices, respectively.
Let $P$ or~$Q$ be $c$-packed for some unknown $c$.
For any $\eps \in (0,1]$, 
\IfRestatedTF{\Call{ApproxFr\'echetDistance}{$P,Q,\eps$}}{there exists an algorithm that}
computes a value~$\Delta$ with $\dF(P,Q)\leq\Delta\leq(1+\eps)\dF(P,Q)$ in time $O{\left(d \cdot c\, \frac{n+m}{\eps} \log{\left(\frac{n + m}{\eps} \right)}\right)}$.
\end{restatable}

Our algorithm runs in near-linear time whenever at least one of the two input curves is $c=O(1)$-packed.
In particular, it matches the running time of the algorithm of Driemel, Har-Peled, and Wenk~\cite{driemel12realistic} up to a logarithmic factor. In contrast to their method, which assumes that \emph{both} curves are $c$-packed, our algorithm requires only one of the two curves to satisfy this condition.
Like the algorithm of~\cite{driemel12realistic}, our approach does not require prior knowledge of the constant $c$. Moreover, it does not require knowing \emph{which} of the two curves is $c$-packed. This robustness makes the algorithm particularly well-suited for practical scenarios, where the structural properties of input data may be unknown. Our algorithm is “almost as fast as possible” compared to conditional lower bounds that rule out running times of $O((d\, c \, n/\sqrt{\varepsilon})^{1-\delta})$ for any $\delta > 0$ when $m \in \Theta(n)$ and the dimension is at least five~\cite{bringmann14hardness}.

\mypar{Overview of techniques.}
By exploring only cells of the $\Delta$-free space that contain at least one $\Delta$-free point, one can easily decide whether the Fr\'echet distance between two curves is at most $\Delta$.
The number of cells with a $\Delta$-free point can be quadratic in general. Yet, for suitably simplified versions of the input curves, we prove in \cref{thm:symmetric_complexity_free_space} the structural property that if at least one of the input curves is $c$-packed, the $(1+\eps)\Delta$-free space has only $O{\left(c\, \frac{(1+\alpha)(n+m)}{\eps} \right)}$ cells that contain a $(1+\eps)\Delta$-free point, where $\alpha=\dF(P, Q)/\Delta$.
This mirrors a key insight by Driemel, Har-Peled, and Wenk~\cite{driemel12realistic}, but has two notable differences: (1) they assume that both curves are $c$-packed, and (2) their bound of $O(c\,\frac{n+m}{\eps})$ does not depend on~$\alpha$.
Our bounds asymptotically match when~$\Delta=\Omega(\dF(P, Q))$.

Driemel, Har-Peled, and Wenk~\cite{driemel12realistic} leveraged their structural insight to define an \emph{approximate decider:} a procedure that, given an approximation parameter $\eps>0$ and a distance threshold $\Delta$, determines whether $\dF(P, Q) > \Delta$ or $\dF(P, Q) \leq (1+\eps)\Delta$.
They used this approximate decider within a binary search framework to compute an overall $(1+\eps)$-approximation.
In contrast, the approximate decider implied by our bound is guaranteed to be efficient only when $\Delta \in \Omega(\dF(P, Q))$.

Even though our approximate decider is potentially inefficient when $\Delta$ is small compared to $\dF(P,Q)$, it can be used to efficiently compute a $(1+\eps)$-approximation of $\dF(P,Q)$.
We show two techniques to achieve this. 
The first method relies on a polynomial over-estimate $\Delta^+$ of $\dF(P, Q)$ (i.e., a value $\Delta^+$ such that $\dF(P, Q) \leq \Delta^+ \leq (n+m)^{O(1)} \dF(P, Q)$).
Various methods exist to compute a such a polynomial over-estimate in near-linear time~\cite{colombe21continuous_frechet,vanderhorst24faster_approximation,vanderhorst23subquadratic_frechet}, but these methods are quite complex.
For the sake of simplicity and self-containment, we additionally show a way to use a straightforward greedy algorithm by Bringmann and Mulzer~\cite{bringmann16approx_discrete_frechet} which computes an exponential over-estimate.

Our first method is careful and calls the approximate decider only with parameters $\Delta$ for which it is efficient; specifically,~$\Delta\geq \dF(P,Q)/2$.
It starts with $\Delta=\Delta^+$, halving $\Delta$ until the decider reports that $\dF(P,Q)>\Delta$.
This results in a range $[\Delta,2(1+\eps)\Delta]$ that contains $\dF(P,Q)$, which, using a similar approach, can further be refined to obtain a $(1+\eps)$-approximation.
We note that if $\Delta^+$ is a polynomial overestimation, $\Delta^+ \leq (n+m)^{O(1)} \cdot \dF(P, Q)$, then the first part of this process makes only $O(\log (n+m))$ calls to our approximate decider. If we instead rely on the simpler to compute exponential over-estimation of $\dF(P, Q)$ then this method is no longer efficient. 

Our second method is a more careless approach, which may call the approximate decider with values~$\Delta$ that are significantly less than $\dF(P, Q)$ -- resulting in an extremely slow decision.
In particular, let $\Delta^+ \leq 2^{O(n+m)} \dF(P, Q)$ and define $\Delta_i := \Delta^+ / (1+\eps)^i$.
Starting with $i=1$, repeatedly double $i$ and call the decider until it reports that $\dF(P,Q)>\Delta_i$.
This will happen after only $O(\log \log_{1+\eps} 2^{O(n+m)}) = O(\log \frac{n+m}{\eps})$ calls to the decider.
This approach returns a range $[\Delta_{2^i}, (1+\eps)\Delta_{2^{i-1}}]$ that contains $\dF(P, Q)$, which can further be refined to obtain a $(1+\eps)$-approximation.
The complication with this approach is that it may invoke our approximate decider with values $\Delta$ that are considerably smaller than $\dF(P, Q)$ -- thereby losing our guarantees on the complexity of the free space.
To work around this, we create a \emph{fallible approximate decider}.
Given estimates $(c', \Delta)$ of $(c, \dF(P, Q))$, this decider may report one of three things: (i) $\dF(P, Q) > \Delta$, (ii) $\dF(P, Q) \leq (1+\eps)\Delta$, or (iii) \emph{timeout}, indicating that the current estimates are too inaccurate.
Note that this decider requires not only an estimate of the \f distance, but also an estimate of the minimum packedness of the curves.
The timeout ensures that the fallible approximate decider has a running time of $O{\left(c \frac{n+m}{\eps}\right)}$ (as long as $c' \in O(c)$) regardless of~$\Delta$.
We use timeouts to guide updates to either $c'$ or $\Delta$, refining our estimates until we obtain a $(1+\eps)$-approximation of $\dF(P, Q)$.

\section{On simplifications and their free spaces}

Driemel, Har-Peled and Wenk~\cite{driemel12realistic} define the following greedy curve simplification of a curve $P$:

\begin{definition}[Greedy $\mu$-simplifications]
\label{def:greedy_simplification}
    \label{def:simplification_driemel}
    Let $\mu > 0$.
    The \emph{greedy $\mu$-simplification} of a curve $P = (p_1, \ldots, p_n)$ is
    the unique maximal subsequence $S = (s_1,\ldots, s_m)$ of $P$ such that $s_1=p_1$, and each~$s_j$ is the first vertex of $P$ after $s_{j-1}$ with distance at least $\mu$ from $s_{j-1}$.
\end{definition}

(Unlike~\cite{driemel12realistic}, we do not necessarily include the last vertex of $P$.
This does not affect their results.)
This class of simplifications has some particularly useful properties.
The simplifications are easily constructed in linear time.
Additionally, Driemel~\etal~\cite{driemel12realistic} show that the simplification retains (roughly) the packedness of the original curve, and is close to the original curve in terms of \f distance:

\begin{lemma}[Lemma~4.3 in~\cite{driemel12realistic}]
    \label{lem:simplificationIsClose}
    Let $P$ be $c$-packed in $\R^d$ with $n$ vertices.
    For any $\mu > 0$, the greedy $\mu$-simplification $S$ of~$P$ can be constructed in $O(dn)$ time.
    The curve $S$ is $6c$-packed, and $\dF(P, S) \leq \mu$.
\end{lemma}

Their primary result, which is central to existing algorithms for $c$-packed curves, is a structural lemma that bounds the complexity of the free space between the two simplified $c$-packed curves:

\begin{lemma}[Lemma~4.4 in~\cite{driemel12realistic}]
    \label{lem:symmetricFreespace}
    Let $P$ and $Q$ be two $c$-packed curves with $n$ and $m$ vertices, respectively.
    Let $\Delta \geq 0$ and $\eps \in (0, 1)$, and let $S_P$ and $S_Q$ be the greedy $\frac{\eps}{4} \Delta$-simplifications of $P$ and $Q$, respectively. Then only $O(c\, \frac{n+m}{\eps})$ parameter space cells of $S_P$ and $S_Q$ intersect $\FS[((1+\frac{\eps}{2})\Delta)](S_P, S_Q)$.
\end{lemma}

We extend \Cref{lem:symmetricFreespace} to the asymmetric setting, where only one curve is $c$-packed:

\begin{restatable}{thm}{symmetricFreespace}\label{thm:symmetric_complexity_free_space}
    Let $P$ and $Q$ be curves with $n$ and $m$ vertices, respectively.
    Let $\Delta > 0$ and $\eps \in (0, 1]$, and let $S_P$ and $S_Q$ be the greedy $\frac{\eps}{4}\Delta$-simplifications of $P$ and $Q$, respectively.
    If one of $P$ and $Q$ is $c$-packed, then at most $\left(27c + 48c \cdot \frac{1+\alpha}{\eps}\right) (n+m) = O{\left(c\, \frac{(1+\alpha)(n+m)}{\eps} \right)}$ parameter space cells of $S_P$ and $S_Q$ intersect $\FS[((1+\frac{\eps}{2})\Delta)](S_P, S_Q)$, where $\alpha = \dF(P, Q) / \Delta$.
\end{restatable}

In the remainder of this section, we assume $P$ to be a $c$-packed curve with $n$ vertices and $Q$ to be an arbitrary curve with $m$ vertices.
We prove an asymmetric bound on the complexity of the free space, which then implies the bound of~\cref{thm:symmetric_complexity_free_space}.
Specifically, let $\Delta$ be given, and let $S$ be the greedy $\frac{\eps}{4}\Delta$-simplification of $P$. By~\cref{lem:simplificationIsClose}, $S$ is $6c$-packed.
We prove that at most $O{\left(c\, \frac{n+(1+\alpha)m}{\eps}\right)}$ parameter space cells of $S$ and $Q$ contain a $(1+\frac{\eps}{2})\Delta$-free point, where $\alpha = \dF(P, Q) / \Delta$.

\subsection{Partitioning \texorpdfstring{$Q$}{Q}.}

To prove our claimed bound on the complexity of the free space of $S$ and $Q$, we split the set of edges of $Q$ into two sets. The first set contains \emph{short} edges, whose arc length is at most $4\dF(P,Q)$. The second set contains the remaining \emph{long} edges, whose arc length is more than $4\dF(P,Q)$. Similarly to the edges, we call a cell $C_{i, j}$ of the parameter space of $S$ and $Q$ \emph{short} if its corresponding edge $Q[j, j+1]$ on $Q$ is short, and call $C_{i, j}$ \emph{long} otherwise.

We show that long edges inherit the $c$-packedness from $P$ (i.e., the total arc length of long edges in any disk of radius $r$ is at most $O(cr)$), while short edges are, in spirit, centers of disks, which can not be intersected by too many edges of $S$, as $S$ is $O(c)$-packed. Thus, there are at most $O{\left( c\, \frac{n+m}{\eps} \right)}$ pairs of edges, one from $S$ and one from $Q$, that are close to one another, and hence there are at most $O{\left( c\, \frac{n+m}{\eps} \right)}$ cells in the parameter space of $S$ and $Q$ that contain a $\frac{\eps}{2}\Delta$-free point.

\begin{figure}
    \centering
    \includegraphics{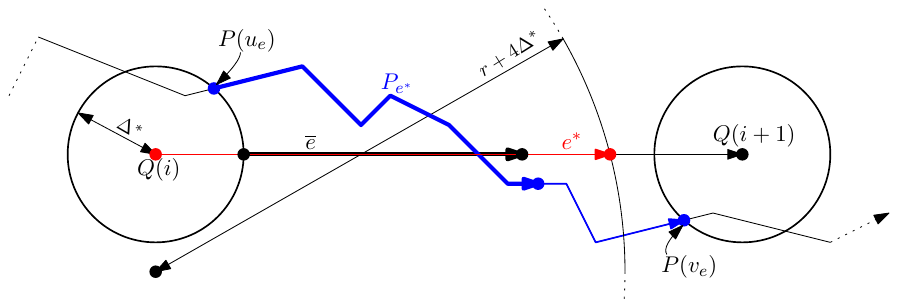}
    \caption{A subcurve $P_e=P[u_e,v_e]$ (in blue) is the minimal subcurve matched to $Q[i,i+1]$ under a matching attaining a distance of $\Delta^*$. The orthogonal projection of the matched subcurve $P_{e^*}$ (in fat blue) of $P_e$ to $e^*$ (in red) covers the segment $\overline{e}$ (in fat) entirely, and is contained in $B(p,r+4\Delta^*)$.}
    \label{fig:orthProj}
\end{figure}

\begin{lemma}
\label{lem:long_edges_packed}
    The set of long edges of $Q$ is $6c$-packed.
\end{lemma}
\begin{proof}
    This proof extends the proofs of \cite[Lemma 4.2 and 4.3]{driemel12realistic} and it is illustrated by \cref{fig:orthProj}.
    
    Take an arbitrary ball $B(p, r)$ and let $E$ denote the set of long edges of $Q$ that intersect $B(p, r)$.
    We show that the length of $E$ inside $B(p, r)$ is at most $6cr$.
 
    Let $\Delta^* = \dF(P, Q)$ and fix a $\Delta^*$-matching $(f,g)$ between $P$ and $Q$.
    For each edge $e=Q[i,i+1] \in E$, let $P_e=P[u_e,v_e]$ denote the minimal subcurve of $P$ matched to $e$ by $(f,g)$. More precisely, $u_e$ is the largest value such that there is a $t\in[0,1]$ such that $f(t)=u_e$ and $g(t)=i$, and $v_e$ is the smallest value such that there is a $t\in[0,1]$ such that $f(t)=v_e$ and $g(t)=i+1$. By minimality, all $[u_e,v_e]$ are disjoint except for possibly at their end points.

    Let us first observe that $\| e \cap B(p, r) \| \leq 2 \cdot \| P_e \cap B(p, r+4\Delta^*) \|$. See \cref{fig:orthProj}.
    Denote by $e^*$ the intersection between $e$ and the ball $B(p, r+4\Delta^*)$.
    Since $e$ intersects $B(p, r)$, we have that $e^*$ is a segment of length at least $4\Delta^*$.
    Let $\overline{e}$ be the segment $e^*$ after truncating it by $\Delta^*$ on either side.
    We have $\|\overline{e}\| = \|e^*\| - 2\Delta^* \geq \|e^*\| / 2$.
    Additionally, the segment $\overline{e}$ lies inside the ball $B(p, r+3\Delta^*)$.

    Let $P_{e^*}$ denote the subcurve of $P_e$ that is matched to $e^*$ by $(f, g)$.
    The orthogonal projection of $P_{e^*}$ onto the line supporting $e^*$ covers the segment $\overline{e}$ entirely.
    Moreover, because the orthogonal projection maps each point to its closest point on the supporting line of $e^*$, we additionally have that any point of $P_{e^*}$ that projects onto $\overline{e}$ lies within distance $\Delta^*$ of $\overline{e}$, and hence lies in the ball $B(p, r+4\Delta^*)$.
    The total length of the parts of $P_{e^*}$ that project onto $\overline{e}$ is at least $\|\overline{e}\| \geq \|e^*\| / 2 = \|e \cap B(p, r) \| / 2$, and as established lies inside $B(p,r+4\Delta^*)$.
    Thus we conclude that $\| e \cap B(p, r) \| \leq 2 \cdot \| P_e \cap B(p, r+4\Delta^*) \|$.

    Recall that all $[u_e, v_e]$ are disjoint, except for possibly at their endpoints.
    Hence, the total length of $P$ inside $B(p, r+4\Delta^*)$ is at least $\sum_{e \in E} \| P_e \cap B(p, r+4\Delta^*) \| \geq \| E \cap B(p, r) \| / 2$.
    It follows from the $c$-packedness of $P$ that $\| E \cap B(p, r) \| \leq 2c \cdot (r+4\Delta^*)$.

    If $r \geq 2\Delta^*$, then $2c \cdot (r+4\Delta^*) \leq 6cr$, proving that the length of $E$ inside $B(p, r)$ is at most $6cr$.
    If $r < 2\Delta^*$ instead, then observe that every edge in $E$ contributes at least $4\Delta^*$ to the length of $E$ inside $B(p, r+4\Delta^*)$, which is at most $2c \cdot (r+4\Delta^*) < 12c\Delta^*$ by the above.
    Hence $E$ contains at most $12c\Delta^* / 4\Delta^* = 3c$ edges.
    Each edge in $E$ contributes at most $2r$ to the length of $E$ inside $B(p, r)$.
    Thus, the length of $E$ inside $B(p, r)$ is at most $6cr$.
\end{proof}


\begin{lemma}
\label{lem:few_intersections_ball}
    Let $c \geq 0$ and $\ell \geq 0$.
    Let $E$ be a $c$-packed set of line segments with lengths at least $\ell$.
    The number of segments in $E$ that intersect a given ball of radius $r$ is at most $c \cdot (1+r/\ell)$.
\end{lemma}
\begin{proof}
    Fix a ball $B(p,r)$.
    Any segment in $E$ that intersects $B(p, r)$, intersects $B(p, r+\ell)$ in a segment of length at least $\ell$.
    Since $E$ is $c$-packed, at most
    \(
        c \cdot \frac{r+\ell}{\ell} = c \cdot (1+r/\ell)
    \)
    such segments exist.
\end{proof}

\subsection{Bounding the \texorpdfstring{$(1+\frac{\eps}{2})\Delta$-}{}free space complexity. }

To simplify notation, let $\Delta_S = (1+\frac{\eps}{2})\Delta$.
\Cref{lem:free_short_cells,lem:free_long_cells} bound how many short and long cells intersect the $\Delta_S$-free space of $S$ and~$Q$.

\begin{lemma}
\label{lem:free_short_cells}
    At most $(12c + 24c / \eps + 48c\alpha / \eps) \cdot m$ short cells contain a $\Delta_S$-free point.
\end{lemma}
\begin{proof}
    Consider a short edge $\overline{q_j q_{j+1}}$ of $Q$.
    For any $i$, the cell $C_{i, j}$ contains a $\Delta_S$-free point if and only if the edge $\overline{s_i s_{i+1}}$ of $S$ contains a point within distance $\Delta_S$ of some point on $\overline{q_j q_{j+1}}$.
    Let $\Delta^* = \dF(P, Q)$ and consider the ball $B(q, 2\Delta^* + \Delta_S)$, centered at the midpoint $p$ of $\overline{q_j q_{j+1}}$.
    This ball contains all points within distance $\Delta_S$ of $\overline{q_j q_{j+1}}$.
    By \Cref{def:greedy_simplification}, all edges of $S$ have length at least $\frac{\eps}{4}\Delta$. We obtain from~\cref{lem:few_intersections_ball} that at most
    \[
        6c \cdot \left(1 + \frac{2\Delta^* + \Delta_S}{\frac{\eps}{4}\Delta} \right) \leq 6c \cdot \left(1 + \frac{2\alpha\Delta + (1+\frac{\eps}{4}\Delta)}{\frac{\eps}{4}\Delta} \right)
        \leq 6c \cdot (2 + 4/\eps + 8\alpha / \eps)
    \]
    edges of $S$ intersect $B(q, 2\Delta^* + \Delta_S)$.
    Thus, there are at most $12c + 24c / \eps + 48c\alpha / \eps$ short cells $C_{i, j}$ containing a $\Delta_S$-free point.
    Summing over all edges of $Q$ proves the claim.
\end{proof}

\begin{lemma}
\label{lem:free_long_cells}
    Fewer than $(15c + 24c / \eps) \cdot (n+m)$ long cells contain a $\Delta_S$-free point.
\end{lemma}
\begin{proof}
    Consider a long cell $C_{i, j}$ that contains a $\Delta_S$-free point.
    Its corresponding edges $\overline{s_i s_{i+1}}$ and $\overline{q_j q_{j+1}}$ contain a pair of points that are within distance $\Delta_S$ of each other.
    We charge the cell $C_{i, j}$ to the shorter of its corresponding edges.
    We claim that no edge can be charged too often.

    Let $E$ be the set of long edges of $Q$.
    Let $u$ be an edge of either $S$ or $E$.
    Any edge $v$ that charges $u$ has length at least $\max\{\|u\|, \frac{\eps}{4}\Delta\}$, since one of $u$ and $v$ is an edge of~$S$, and therefore has length at least $\frac{\eps}{4}\Delta$.
    Let $p$ be the midpoint of $u$ and consider the ball $B(p, r)$ of radius $r = \|u\| / 2 + \Delta_S$.
    Every edge that charges $u$ intersects $B(p, r)$.
    By~\cref{lem:simplificationIsClose,lem:long_edges_packed}, both $S$ and $E$ are $6c$-packed.
    We therefore obtain from~\cref{lem:few_intersections_ball} that at most
    \[
        6c \cdot \left(1 + \frac{\|u\|/2 + \Delta_S}{\max\{\|u\|, \frac{\eps}{4}\Delta\}} \right) < 6c \cdot \left( 1 + \frac{1}{2} + \frac{(1+\frac{\eps}{4})\Delta}{\frac{\eps}{4} \Delta} \right) \leq 15c + 24 c / \eps
    \]
    edges of $S$, and similarly $E$, that are longer than $u$ intersect $B(p, r)$.
    Thus, $u$ is charged fewer than $15c + 24c / \eps$ times.
    The claim follows by summing over all edges of $S$ and long edges of $Q$.
\end{proof}

From \cref{lem:free_short_cells,lem:free_long_cells}, it follows that at most $(15c + 24c/\eps) \cdot n + (27c + 48 c / \eps + 48 c\alpha / \eps) \cdot m$ parameter space cells of $S$ and $Q$ intersect $\FS[(1+\frac{\eps}{2})\Delta](S, Q)$.
Because the greedy simplification of a curve has at most as many vertices of the original curve, we obtain~\cref{thm:symmetric_complexity_free_space}:

\symmetricFreespace*

\section{Approximating the Fr\'echet distance when one curve is \texorpdfstring{$c$}{c}-packed}

We present two algorithms for computing a $(1+\eps)$-approximation of the \f distance between two curves $P$ and $Q$.
Let $P$ have $n$ vertices and $Q$ have $m$ vertices, and suppose one of $P$ and $Q$ is $c$-packed for some unknown value $c$.
Let $\eps \in (0, 1]$ denote an approximation parameter.

Let $\Delta$ denote any estimate of~$\dF(P, Q)$ and define $\alpha := \frac{\dF(P, Q)}{\Delta}$.
Let $S_P$ and $S_Q$ be the greedy $\frac{\eps}{4}\Delta$-simplifications of $P$ and $Q$, respectively.
We proved in~\cref{thm:symmetric_complexity_free_space} that the number of cells in the parameter space of $S_P$ and $S_Q$ that contain a $(1+\frac{\eps}{2})\Delta$-free point is at most
\[
    K_{c}(\Delta) := \left(27c + 48c \cdot \frac{1+\alpha}{\eps}\right) (n+m).
\]
Our overarching approach is now straightforward:
We compute $S_P$ and $S_Q$ in linear time. 
If either $P$ or $Q$ is $c$-packed, the argument above guarantees that at most $K_c(\Delta)$ cells in the parameter space of $S_P$ and $S_Q$ are $\Delta$-reachable in $\FS[((1+\frac{\eps}{2})\Delta)](S_P, S_Q)$.
Thus, by exploring the entire $\Delta$-reachable subset of $\FS[((1+\frac{\eps}{2})\Delta)](S_P, S_Q)$, we can decide whether $\dF(P, Q) \leq (1+\eps)\Delta$ or $\dF(P, Q) > \Delta$ in $O(d K_{c}(\Delta))$ time.
We present two algorithms that use this principle to compute a $(1+\eps)$-approximation for $\dF(P, Q)$. 
Our first algorithm is the most simple. It starts with a polynomial over-estimation of $\dF(P, Q)$. That is, a value $\Delta^+$ such that $\dF(P, Q) \leq \Delta^+ \leq (n+m)^{O(1)}$. 
There exist near-linear algorithms to compute such an approximation, see~\cite{colombe21continuous_frechet,vanderhorst24faster_approximation, vanderhorst23subquadratic_frechet}. However, these methods are involved.
Our second approach is more involved, but it is more self-contained in the sense that it uses an exponential over-estimation of $\dF(P, Q)$ instead, which can be computed with a straightforward greedy algorithm (originally from~\cite{bringmann16approx_discrete_frechet}).

\subsection{A simple packedness-oblivious algorithm.}

Our first algorithm (Algorithm~\ref{alg:polynomial}) starts with an over-estimate $\Delta^+$ of $\dF(P, Q)$. We aggressively halve this interval until $[\Delta^+/2, \Delta^+]$ is a constant range containing $\dF(P, Q)$.
We then binary search over this interval.
Notably, our algorithm does not explicitly use the quantity $K_{c}(\Delta)$. We simply guarantee through $K_c(\Delta)$ that each iteration takes $O{\left( c \, \frac{n+m}{\eps} \right)}$ time.

\begin{algorithm}[H]\caption{$(1+\eps)$-approximation algorithm for the Fréchet distance given an over-estimate $\Delta^+$.
%
}
\label{alg:polynomial}
\begin{algorithmic}[1]

\Procedure{ApproxFréchetDistance}{$P,Q,\Delta^+, \eps$}
    \State $\eps' \gets \eps / 3$
    \State Linearly scan for the smallest $b\in\mathbb{N}$, such that $\Call{ApproxDecider}{P,Q,\Delta^+/2^b,\eps'}\neq\Yes$
    \State Let $\Delta_0=\Delta^+/2^{b-1}$, and denote by $\Delta_i$ the value $\Delta_0/(1+\eps')^i$
    \State    Binary search for $i\in[\lceil\log_{1+\eps'}2\rceil]$, such that $\Call{ApproxDecider}{P,Q,\Delta_{i-1},\eps'} = \Yes$\\
    \hphantom{Binary search for $i\in[\lceil\log_{1+\eps'}2\rceil]$, such that}\neghphantom{ and} and $\Call{ApproxDecider}{P,Q,\Delta_{i\hphantom{-1}},\eps'} = \No$
    \State \Return $(1+\eps')\Delta_{i-1}$
\EndProcedure
\Statex\vspace{-1.5ex}
\Procedure{ApproxDecider}{$P,Q,\Delta, \eps'$}
    \State $(S_P, S_Q) \gets$ the greedy $\frac{\eps}{4} \Delta$-simplifications of $P$ and $Q$
    \State \textbf{if} $\dF(S_P,S_Q)\leq (1+\frac{\eps'}{2})\Delta$ \textbf{then} \Return \Yes \textbf{else} \Return \No
\EndProcedure
\end{algorithmic}
\end{algorithm}

\begin{lemma}
\label{lemm:with_polynomial_estimate}
Let $P$ and $Q$ be polygonal curves with $n$ and $m$ vertices, respectively.
Let $P$ or $Q$ be $c$-packed for some unknown $c$.
Given values $\Delta^+$ with $\dF(P, Q) \leq \Delta^+$ and $\eps \in (0, 1]$, Algorithm~\ref{alg:polynomial} computes a value $\Delta$ with $\dF(P, Q) \leq \Delta < (1+\eps)^2\dF(P,Q)$ in $O{\left(d \cdot c\, \frac{n+m}{\eps}\left(1+\log \left( \frac{\Delta^+}{\eps \cdot \dF(P, Q)}\right) \right) \right)}$ time.
\end{lemma}

\begin{proof}
$\Call{ApproxDecider}{P, Q, \Delta, \eps'}$ either decides that $\dF(P, Q) \leq (1+\eps') \Delta$ or $\dF(P, Q) > \Delta$.
It follows that $\Delta_0$ is a $2(1+\eps')$-approximation of $\dF(P, Q)$.
Via the same argument, the algorithm outputs, after the binary search for $i$, a $(1+\eps')^2$-approximation.
This is a $(1+\eps)$-approximation since $(1+\eps')^2 = (1+\frac{\eps}{3})^2 \leq 1+\eps$.
It remains to analyse the running time.
Since we decrease $b$ until $\frac{\Delta^+}{2}$ is smaller than $\dF(P, Q)$, this scan has at most $O(\log \frac{\Delta^+
}{\dF(P, Q)})$ iterations. 
Binary search for $i \in \lceil \log_{1+\eps} 2 \rceil$ takes $O(\log \log_{1+\eps} 2) = O(\log \eps^{-1})$ iterations. 
Finally, our approach guarantees that ApproxDecider is always invoked with a value $\Delta \geq \frac{\dF(P, Q)}{2}$.
By~\cref{thm:symmetric_complexity_free_space}, the free-space has at most $O(K_c(\Delta))$ $\Delta$-reachable cells.
The value $\alpha$ in the function of $K_c(\Delta)$ is, by our choice of $\Delta$, at most two.
We compute the $\Delta$-free space inside any cell in $O(d)$ time, resulting in a running time of $O(d K_c(\Delta))$.
\end{proof}

This result can be combined with a linear approximation for $\dF(P, Q)$.
E.g., Colombe and Fox~\cite[Corollary~4.4]{colombe21continuous_frechet}, which runs in $O(d^2 \cdot (n+m) \log (n+m)) + 2^{O(d)} (n+m)$ time.
This results in \cref{thm:weaker_main}, which is weaker than~\cref{thm:main} because of the super-linear dependency on the dimension~$d$.

\begin{thm}\label{thm:weaker_main}
    Let $P$ and $Q$ be polygonal curves in $\R^d$ with $n$ and $m$ vertices, respectively.
    Let $P$ or $Q$ be $c$-packed for some unknown $c$.
    For any $\eps \in (0, 1]$, there exists an algorithm that computes a value $\Delta$ with $\dF(P, Q) \leq \Delta \leq (1+\eps)\dF(P, Q)$ in time $O{\left(d \cdot c\, \frac{n+m}{\eps} \log \left(\frac{n+m}{\eps}\right) + d^2 \cdot (n+m) \log (n+m) \right)} + 2^{O(d)} (n+m)$.
\end{thm}

\subsection{A more self-contained algorithm.}

Finally, we show an alternative algorithm, to prove Theorem~\ref{thm:main}. 
In particular, we use the straightforward greedy algorithm by Bringmann and Mulzer~\cite{bringmann16approx_discrete_frechet} to compute an exponential overestimation $\Delta^+$ of $\dF(P, Q)$, rather than a more complex linear over-estimation.
This makes our approach more self-contained, as we can fully specify all required algorithms.

We define what we call a \emph{fallible decider} (subroutine \textsc{FallibleDecider} in~\cref{alg:approximation}).
Intuitively, our fallible decider receives as input an estimate $\Delta$ of $\dF(P, Q)$ and an estimate $c'$ of the $c$-packedness of~$P$. 
It computes the greedy $\frac{\eps}{4}\Delta$-simplifications $S_P$ and $S_Q$ of $P$ and $Q$, and explores the parameter space of $S_P$ and $S_Q$ in search of a monotone path from $(1, 1)$ to $(|S_P|, |S_Q|)$.
However, we terminate the search early and report \Timeout if it explores more than the following number of cells:
\[
    K_{c'} := (1+\eps)^2 \left(27c' + \frac{96c'}{\eps}\right) (n+m).
\]
This keeps the running time low.
Observe that $K_{c'} \geq K_c(\Delta)$ whenever $c' \geq c$ and $\Delta \geq \dF(P, Q)$.

\begin{lemma}
\label{lem:fallible_decider}
    \emph{$\Call{FallibleDecider}{P, Q, \Delta, \eps, c'}$} runs in $O{\left(d \cdot c'\, \frac{n+m}{\eps} \right)}$ time.
    If it reports \Yes, then $\dF(P, Q) \leq (1+\eps)\Delta$.
    If it reports \No, then $\dF(P, Q) > \Delta$.
    If it reports \Timeout, then neither $P$ nor $Q$ is $c'$-packed, or $\dF(P, Q) > \Delta$.
\end{lemma}
\begin{proof}
    We first prove correctness.
    Observe that if the algorithm does not report \Timeout, the algorithm has either found a monotone path from $(1, 1)$ to $(|S_P|, |S_Q|)$, or has determined that no such path exists.
    In the former case, the path serves as a witness that $\dF(S_P, S_Q) \leq (1+\eps/2)\Delta$, and by the triangle inequality, $\dF(P, Q) \leq (1+\eps/2)\Delta + \eps\Delta / 4 + \eps\Delta / 4 = (1+\eps)\Delta$.
    In the latter case, we have $\dF(S_P, S_Q) > (1+\eps/2)\Delta$, and by the triangle inequality, $\dF(P, Q) > \Delta$.
    
    Suppose it does report \Timeout.
    Then, more than $K_{c'}$ cells of the parameter space of $S_P$ and $S_Q$ contain a $(1+\eps/2)\Delta$-free point.
    However, we obtain from~\cref{thm:symmetric_complexity_free_space} that if $P$ or $Q$ is $c'$-packed and $\Delta \geq \dF(P, Q)$, at most $K_{c'}(\Delta) \leq (27c' + 96'/\eps) (n+m) \leq K_{c'}$ cells contain a $(1+\eps/2)\Delta$-free point.
    Thus, either $P$ and $Q$ are both not $c'$-packed, or $\dF(P, Q) > \Delta$ (or both).

    We now bound the running time.
    The greedy $\frac{\eps}{4}\Delta$-simplifications $S_P$ and $S_Q$ of $P$ and~$Q$ are constructed in $O(d (n+m))$ time~\cite{driemel12realistic}.
    Analogous to Breadth-First Search, we maintain a queue of unexplored reachable cells, and compute the reachability of any cell in $O(d)$ time.
    Our algorithm explores at most $K_{c'}$ parameter space cells, so the total running time is $O(d K_{c'})$.
\end{proof}

We use our fallible decider in a search procedure to approximate $\dF(P, Q)$.
The search procedure is similar to~\cref{alg:polynomial}, although it requires some care when using \textsc{FallibleDecider} instead of a proper decision algorithm.
We list the search procedure in~\cref{alg:approximation} (subroutine \textsc{RefineUpperBound}).

Like~\cref{alg:polynomial}, the algorithm starts with an over-estimate $\Delta^+$ of $\dF(P, Q)$.
However, due to \textsc{FallibleDecider} requiring an estimate of the packedness of $P$, we also use an estimate $c' := 1$ of~$c$ that we update throughout the algorithm, while maintaining the invariant that $c' < 2c$.
For every considered estimate $c'$, the algorithm first computes a sufficiently small interval that contains $\dF(P, Q)$.
We then perform a binary search over this interval.

\begin{lemma}
    Let $P$ and $Q$ be polygonal curves with $n$ and $m$ vertices, respectively.
    Let $P$ or $Q$ be $c$-packed for some unknown $c$.
    For arguments $\Delta^+\geq \dF(P, Q)$ and $\eps \in (0, 1]$, \emph{\textsc{RefineUpperBound}} computes a value~$\Delta$ with $\dF(P, Q) \leq \Delta < (1+\eps)^2 \dF(P, Q)$ in $O{\left(d \cdot c\, \frac{n+m}{\eps} \left(1 + \log \left(1 + \frac{1}{\eps} \log \frac{\Delta^+}{\dF(P, Q)} \right) \right) \right)}$~time.
\end{lemma}
\begin{proof}
    We prove the correctness of \textsc{RefineUpperBound} for any estimate $c'$ of $c$.
    For brevity, we omit the arguments $P$ and $Q$ passed to \textsc{FallibleDecider}, as these do not vary.
    By~\cref{lem:fallible_decider}, if one of $P$ and $Q$ is $c'$-packed, the call $\Call{FallibleDecider}{\Delta^+, \eps, c'}$ in line~\ref{line:check_over-estimate} reports \Yes, since $\dF(P, Q) \leq \Delta^+$.
    Let $b \in \mathbb{N}$ be the smallest integer such that $\Call{FallibleDecider}{\Delta_{2^b}, \eps, c'}$ does not report \Yes.
    Then there exists an integer $i \in [0, 2^b]$ such that $\Call{FallibleDecider}{\Delta_{i-1}, \eps, c'}$ reports \Yes, but $\Call{FallibleDecider}{\Delta_i, \eps, c'}$ does not.
    By~\cref{lem:fallible_decider}, we have $\dF(P, Q) \leq (1+\eps)\Delta_{i-1}$, and additionally $\dF(P, Q) > \Delta_i$ (meaning $\dF(P, Q) > \Delta_{i-1} / (1+\eps)$) or neither $P$ nor $Q$ is $c'$-packed.

    We analyse line~\ref{line:check_result}, in which we decide whether to report $(1+\eps)\Delta_{i-1}$ as a $(1+\eps)^2$-approximation to $\dF(P, Q)$.
    Suppose $P$ or $Q$ is $c'$-packed.
    We show that $\Call{FallibleDecider}{\Delta_i, \eps, c'}$ will not report \Timeout, and thus reports \No.
    Let $\alpha_i = \dF(P, Q) / \Delta_i$ and $\alpha_{i-1} = \dF(P, Q) / \Delta_{i-1}$.
    Then
    \[
        K_{c'}(\Delta_i) = \left(27c' + 48c'\, \frac{1+\alpha_i}{\eps} \right) (n+m) \leq (1+\eps) \left(27c' + 48c'\, \frac{1+\alpha_{i-1}}{\eps} \right) (n+m) = (1+\eps) K_{c'}(\Delta_{i-1}).
    \]
    Observe that $K_{c'}(\Delta_{i-1}) \leq K_{c'} / (1+\eps)$, since $\dF(P, Q) \leq (1+\eps)\Delta_{i-1}$ and thus $\alpha_{i-1} \leq 1+\eps$.
    Hence $K_{c'}(\Delta_i) \leq (1+\eps) K_{c'}(\Delta_{i-1}) \leq K_{c'}$, and the call $\Call{FallibleDecider}{\Delta_i, \eps, c'}$ will not report \Timeout.
    This concludes the proof of correctness.

    Next we analyse the running time for a fixed estimate $c'$ of $c$.
    We bound the number of calls to \textsc{FallibleDecider}, which dominate the running time.
    The linear scan on line~\ref{line:linear_scan} terminates at or before the first value $b^* \in \mathbb{N}$ with $\Delta_{2^{b^*}} < \dF(P, Q)$.
    It follows from our definition of $\Delta_{2^{b^*}} = \Delta^+ / (1+\eps)^{2^{b^*}}$ that $b^*$ is the minimum value for which $2^{b^*} > \log_{1+\eps} \frac{\Delta^+}{\dF(P, Q)}$.
    Given that $b^*$ is an integer, we obtain $2^{b^*} \geq \lfloor \log_{1+\eps} \frac{\Delta^+}{\dF(P, Q)} \rfloor + 1$.
    We note that since $\Delta^+ \geq \dF(P, Q)$, this lower bound on $2^{b^*}$ is at least $1$.
    In particular, we can take the logarithm on both sides to obtain that $b^* = \left\lceil \log_2 \left(1 + \left\lfloor \log_{1+\eps} \frac{\Delta^+}{\dF(P, Q)} \right\rfloor \right) \right\rceil = O{\left(\log_2 \left( 1 + \log_{1+\eps} \frac{\Delta^+}{\dF(P, Q)} \right) \right)}$.
    In line~\ref{line:binary_search} we perform a binary search over $[2^b]$, for the value $b \leq b^*$ found in line~\ref{line:linear_scan}.
    This search performs at most $b^*$ steps.
    Thus, both lines~\ref{line:linear_scan} and~\ref{line:binary_search} perform at most $b^* = O{\left(1 + \log_2 \left( 1 + \log_{1+\eps} \frac{\Delta^+}{\dF(P, Q)} \right) \right)}$ calls to \textsc{FallibleDecider}.

    Factoring in the running time of \textsc{FallibleDecider} (which is $O(d K_{c'}) = O{\left(d \cdot c'\, \frac{n+m}{\eps} \right)}$), the running time of \textsc{RefineUpperBound} for a fixed estimate $c'$ is $O{\left(d \cdot c'\, \frac{n+m}{\eps} \left(1+\log_2 \left( 1 + \log_{1+\eps} \frac{\Delta^+}{\dF(P, Q)} \right) \right) \right)}$.
    Summed over the values $c'$, which we double after each iteration up to a maximum of $2c-1$, we obtain a total running time of $O{\left(d \cdot c\, \frac{n+m}{\eps} \left(1+\log_2 \left( 1 + \log_{1+\eps} \frac{\Delta^+}{\dF(P, Q)} \right) \right) \right)}$.
    The claim follows from the fact that
    $
        \log_{1+\eps} x = \frac{\log_2 x}{\log_2 (1+\eps)} \leq \frac{\log_2 x}{\eps}, \text{ for all } \eps \in (0, 1]. 
    $
\end{proof}

We may invoke \textsc{RefineUpperBound} with an exponential over-estimate $\Delta^+\leq 2^{O(n+m)} \dF(P, Q)$.
Such an estimate is easily obtained in $O(d (n+m))$ time with the straightforward greedy algorithm by Bringmann and Mulzer~\cite{bringmann16approx_discrete_frechet}, listed as \textsc{ExponentialApprox} in~\cref{alg:approximation}.
Thus,~\cref{alg:approximation} runs in $O{\left(d \cdot c\, \frac{n+m}{\eps} \log{\left(\frac{n+m}{\eps} \right)}\right)}$ time, improving the dependency on $d$ significantly compared to~\cref{alg:polynomial}.
With this, we have a largely self-contained algorithm:

\thmMain*

\newpage

\begin{algorithm}[H]\caption{
    Our $(1+\eps)$-approximation algorithm for the Fréchet distance and its three subroutines.
    Complete Python code is available at \href{https://doi.org/10.5281/zenodo.17309708}{doi.org/10.5281/zenodo.17309708}.
    }
    \label{alg:approximation}
    \begin{algorithmic}[1]
    \Procedure{ApproxFr\'echetDistance}{$P,Q,\eps$}\Comment{Our $(1+\eps)$-approximation algorithm}
        \State $\Delta^+\gets\Call{ExponentialApprox}{P,Q}$
        \State $\eps' \gets \eps/3$ \Comment{$(1+\eps')^2 \leq 1+\eps$}
        \If{$\Delta^+=0$} \Return $0$ \textbf{else} \Return $\Call{RefineUpperBound}{P,Q,\Delta^+,\eps'}$ \EndIf
    \EndProcedure
    
    \Statex\vspace{-1.5ex}
    \Function{RefineUpperBound}{$P,Q,\Delta^+,\eps$}\Comment{Search for $\dF(P,Q)$ while estimating packedness}
    \LComment{Outputs some $\Delta$ such that $\dF(P,Q) \leq \Delta < (1+\eps)^2 \dF(P,Q)$.}
        \State $c'\gets 1$, and denote by $\Delta_i$ the value $\Delta^+/(1+\eps)^i$
        \Loop
            \If{$\Call{FallibleDecider}{P,Q,\Delta^+,\eps,c'} = \Yes$}\label{line:check_over-estimate}
            \State Linearly scan for the smallest $b\in\mathbb{N}$, such that $\Call{FallibleDecider}{P,Q,\Delta_{2^b},\eps,c'}\neq \Yes$\label{line:linear_scan}
            \State Binary search for $i\in[{2^b}]$, such that $\Call{FallibleDecider}{P,Q,\Delta_{i-1},\eps,c'}\!=\!\Yes$ \\
            \hphantom{Binary search for $i\in[{2^b}]$, such that\neghphantom{ and}} and $\Call{FallibleDecider}{P,Q,\Delta_{i\hphantom{-1}},\eps,c'}\!\neq\! \Yes$\label{line:binary_search}
            \If{$\Call{FallibleDecider}{P,Q,\Delta_i,\eps,c'}=\No$}\label{line:check_result}
                \State \Return $(1+\eps)\Delta_{i-1}$
            \EndIf
            \EndIf
                \State $c'\gets 2c'$ \Comment{\textsc{FallibleDecider} reported \Timeout}
        \EndLoop
    \EndFunction
    
    \Statex\vspace{-1.5ex}
    \Function{FallibleDecider}{$P,Q,\Delta,\eps,c$}\Comment{The fuzzy decision oracle guiding binary searches}
    \LComment{\spaceskip=2.3pt plus 2pt 
             Outputs \Yes/\No/\Timeout.
             \Timeout implies that more than $K_\eps$ cells are $\Delta$-reachable.
             \Yes implies that $\dF(P,Q)\leq(1+\eps)\Delta$. 
             \No implies that $\dF(P,Q)>\Delta$.}
        \State $(S_P, S_Q) \gets $ the greedy $\frac{\eps}{4}\Delta$-simplifications of $P$ and $Q$
        \State $K_c\gets (1+\eps)^2 (27c + 96c / \eps) (n+m)$
        \State $R\gets$ reachable points of the first $K_c+1$ reachable cells of $\FS[(1+\frac{\eps}{2})\Delta](S_P, S_Q)$
        \If{$R$ contains $(|S_P|,|S_Q|)$}
            \State
            \Return\Yes \Comment{$\dF(S_P, S_Q) \leq (1+\eps/2)\Delta$, so $\dF(P,Q)\leq (1+\eps)\Delta$}
        \ElsIf{$R$ contains points of more than $K_c$ cells}
            \State
            \Return\Timeout
        \Else\Comment{All reachable points were found, but not $(|S_P|,|S_Q|)$.}
            \State\Return\No \Comment{$\dF(S_P, S_Q) > (1+\eps/2)\Delta$, so $\dF(P,Q)>\Delta$}
        \EndIf
    \EndFunction

    \Statex\vspace{-1.5ex}
    \Function{ExponentialApprox}{$P,Q$}\Comment{Greedy upper bound on $\dF(P,Q)$ from \cite{bringmann16approx_discrete_frechet}}
    \LComment{Outputs some $\Delta^+$ such that $\dF(P,Q)\leq\Delta^+\leq 2^{O(n+m)}\dF(P,Q)$.}
        \State $(x,y)\gets(1,1)$ and $d\gets\|P(1) - Q(1)\|$
        \While{$(x,y)\neq(n,m)$}
            \State Let $P(x')$ be the first vertex strictly after $P(x)$, or $P(n)$ if it does not exist
            \State Let $Q(y')$ be the first vertex strictly after $Q(y)$, or $Q(m)$ if it does not exist
            \State $x^* \gets \arg\min_{x^*\in[x,x']}\|P(x^*) - Q(y')\|$
            \State $y^* \gets \arg\min_{y^*\in[y,y']}\|P(x') - Q(y^*)\|$
            \If{$x=n$ or $y=m$}
                \State $(x,y)\gets(x',y')$
            \ElsIf{$\|P(x^*) - Q(y')\|\leq \|P(x') - Q(y^*)\|$}
                \State $(x,y)\gets(x^*,y')$
            \Else
                \State $(x,y)\gets(x',y^*)$
            \EndIf
            \State $d\gets\max(d,\|P(x)-Q(y)\|)$
        \EndWhile
        \State\Return $d$
    \EndFunction

\end{algorithmic}
\end{algorithm}

\bibliographystyle{plainurl}
\bibliography{bibliography}

@article{driemel12realistic,
  author    = {Anne Driemel and
               Sariel Har{-}Peled and
               Carola Wenk},
  title     = {Approximating the {F}r{\'{e}}chet Distance for Realistic Curves in Near Linear Time},
  journal   = {Discrete \& Computational Geometry},
  volume    = {48},
  number    = {1},
  pages     = {94--127},
  year      = {2012},
  doi       = {10.1007/s00454-012-9402-z},
}

@inproceedings{driemel2023discrete,
  title={On the Discrete {F}r\'{e}chet Distance in a Graph},
  author={Driemel, Anne and van der Hoog, Ivor and Rotenberg, Eva},
  booktitle={Proc. 38th International Symposium on Computational Geometry (SoCG)},
  year={2022},
    doi = {10.4230/LIPIcs.SoCG.2022.36}
}

@inproceedings{Hoog2024DataRealistic,
  author =	{Hoog, Ivor van der and Rotenberg, Eva and Wong, Sampson},
  title =	{Data Structures for Approximate {F}r\'{e}chet Distance for Realistic Curves},
  booktitle =	{Proc. 35th International International Symposium on Algorithms and Computation (ISAAC)},
  year =	{2024},
  volume =	{322},
  doi =		{10.4230/LIPIcs.ISAAC.2024.56}
}

@article{bombelli2017strategic,
  title={Strategic Air Traffic Planning with {F}r{\'e}chet distance aggregation and rerouting},
  author={Bombelli, Alessandro and Soler, Lluis and Trumbauer, Eric and Mease, Kenneth D},
  journal={Journal of Guidance, Control, and Dynamics},
  volume={40},
  number={5},
  pages={1117--1129},
  year={2017},
  publisher={American Institute of Aeronautics and Astronautics},
    doi = {10.2514/1.G002308}
}

@article{jiang2008protein,
  title={Protein structure--structure alignment with discrete {F}r{\'e}chet distance},
  author={Jiang, Minghui and Xu, Ying and Zhu, Binhai},
  journal={Journal of bioinformatics and computational biology},
  volume={6},
  number={01},
  pages={51--64},
  year={2008},
  publisher={World Scientific},
    doi = {10.1142/s0219720008003278}
}

@article{buchin2020group,
  title={Group diagrams for representing trajectories},
  author={Buchin, Maike and Kilgus, Bernhard and K{\"o}lzsch, Andrea},
  journal={International Journal of Geographical Information Science},
  volume={34},
  number={12},
  pages={2401--2433},
  year={2020},
  publisher={Taylor \& Francis},
    doi = {10.1145/3283207.3283208}
}

@inproceedings{brakatsoulas2005map,
  title={On map-matching vehicle tracking data},
  author={Brakatsoulas, Sotiris and Pfoser, Dieter and Salas, Randall and Wenk, Carola},
  booktitle={Proc. International Conference on Very Large Data Bases (VDLP)},
  pages={853--864},
    doi = {10.5555/1083592.1083691},
  year={2005},
    publisher = {Springer}
}

@inproceedings{sriraghavendra2007frechet,
  title={{F}r{\'e}chet distance based approach for searching online handwritten documents},
  author={Sriraghavendra, E and Karthik, K and Bhattacharyya, Chiranjib},
  booktitle={Proc. Ninth International Conference on Document Analysis and Recognition (ICDAR)},
  volume={1},
  pages={461--465},
  year={2007},
  organization={IEEE},
    doi = {10.1109/ICDAR.2007.4378752}
}

@article{su2020survey,
  title={A survey of trajectory distance measures and performance evaluation},
  author={Su, Han and Liu, Shuncheng and Zheng, Bolong and Zhou, Xiaofang and Zheng, Kai},
  journal={The VLDB Journal},
  volume={29},
  number={1},
  pages={3--32},
  year={2020},
  publisher={Springer},
    doi = {10.1007/s00778-019-00574-9}
}

@article{acmsurvey20,
author = {Sousa, Roniel S. De and Boukerche, Azzedine and Loureiro, Antonio A. F.},
title = {Vehicle Trajectory Similarity: Models, Methods, and Applications},
journal = {ACM Computing Surveys},
year = {2020},
publisher = {Association for Computing Machinery},
volume = {53},
number = {5},
doi = {10.1145/3406096},
}

@incollection{devogele2002new,
  title={A new merging process for data integration based on the discrete {F}r{\'e}chet distance},
  author={Devogele, Thomas},
  booktitle={Advances in spatial data handling},
  pages={167--181},
  year={2002},
  publisher={Springer},
    doi = {10.1007/978-3-642-56094-1_13}
}

@incollection{mascret2006coastline,
  title={Coastline matching process based on the discrete {F}r{\'e}chet distance},
  author={Mascret, Ariane and Devogele, Thomas and Le Berre, Iwan and H{\'e}naff, Alain},
  booktitle={Progress in Spatial Data Handling},
  pages={383--400},
  year={2006},
  publisher={Springer},
    doi = {10.1007/3-540-35589-8_25}
}

@misc{harpeled2025packeditreally,
      title={How Packed Is It, Really?}, 
      author={Sariel Har-Peled and Timothy Zhou},
      year={2025},
      eprint={2105.10776},
      archivePrefix={arXiv},
      primaryClass={cs.CG},
      url={https://arxiv.org/abs/2105.10776}, 
}

@article{bringmann17cpacked,
  author       = {Karl Bringmann and
                  Marvin K{\"{u}}nnemann},
  title        = {Improved Approximation for {F}r{\'{e}}chet Distance on $c$-Packed Curves Matching Conditional Lower Bounds},
  journal      = {International Journal of Computational Geometry \& Applications},
  volume       = {27},
  number       = {1-2},
  pages        = {85--120},
  year         = {2017},
  doi          = {10.1142/S0218195917600056},
}

@article{alt95continuous_frechet,
  author    = {Helmut Alt and
               Michael Godau},
  title     = {Computing the {F}r{\'{e}}chet distance between two polygonal curves},
  journal   = {International Journal of Computational Geometry \& Applications},
  volume    = {5},
  pages     = {75--91},
  year      = {1995},
  doi       = {10.1142/S0218195995000064},
}

@article{Driemel2013Jaywalking,
author = {Driemel, Anne and Har-Peled, Sariel},
title = {Jaywalking Your Dog: Computing the {F}réchet Distance with Shortcuts},
journal = {SIAM Journal on Computing},
volume = {42},
number = {5},
pages = {1830-1866},
year = {2013},
doi = {10.1137/120865112}
}

@inproceedings{gudmundsson2024Approximating,
  author =	{Gudmundsson, Joachim and Mai, Tiancheng and Wong, Sampson},
  title =	{Approximating the {F}r\'{e}chet Distance When Only One Curve Is $c$-Packed},
  booktitle =	{Proc. 35th International Symposium on Algorithms and Computation (ISAAC)},
  pages =	{37:1--37:14},
  year =	{2024},
  doi =		{10.4230/LIPIcs.ISAAC.2024.37}
}

@article{CONRADI2025Revisiting,
title = {Revisiting the {F}réchet distance between piecewise smooth curves},
journal = {Computational Geometry},
author = { Conradi, Jacobus and Driemel, Anne and  Kolbe, Benedikt},
volume = {129},
pages = {102194},
year = {2025},
issn = {0925-7721},
doi = {https://doi.org/10.1016/j.comgeo.2025.102194}
}

@inproceedings{gudmundsson2023Cluster,
  author =	{Gudmundsson, Joachim and Huang, Zijin and van Renssen, Andr\'{e} and Wong, Sampson},
  title =	{Computing a Subtrajectory Cluster from $c$-Packed Trajectories},
  booktitle =	{Proc. 34th International Symposium on Algorithms and Computation (ISAAC)},
  pages =	{34:1--34:15},
  year =	{2023},
  doi =		{10.4230/LIPIcs.ISAAC.2023.34},
}

@misc{conradi2025Subtrajectory,
      title={Subtrajectory Clustering and Coverage Maximization in Cubic Time, or Better}, 
      author={Jacobus Conradi and Anne Driemel},
      year={2025},
      archivePrefix={arXiv preprint},
      url={https://arxiv.org/abs/2504.17381}, 
}

@inproceedings{Hoog2025Clustering,
  title={Faster, Deterministic and Space Efficient Subtrajectory
  Clustering},
  author={van der Hoog, Ivor and van der Horst, Thijs and Ophelders, Tim},
  booktitle={Proc. 52nd International Colloquium on Automata, Languages, and Programming (ICALP) },
  year={2025},
  pages = {133:1--133:18},
  doi = {10.4230/LIPIcs.ICALP.2025.133},
}

@inproceedings{buchin2014shortcutNP,
  title={Computing the {F}r{\'e}chet distance with shortcuts is NP-hard},
  author={Buchin, Maike and Driemel, Anne and Speckmann, Bettina},
  booktitle={Proc. Thirtieth Annual Symposium on Computational Geometry (SoCG)},
  pages={367--376},
  year={2014}
}

@article{conradi2024kShortcut,
  title={On Computing the $k$-Shortcut {F}r{\'e}chet Distance},
  author={Conradi, Jacobus and Driemel, Anne},
  journal={ACM Transactions on Algorithms},
  volume={20},
  number={4},
  pages={1--37},
  year={2024},
}

@article{conradi20231+,
  title={(1+$\varepsilon$)-{ANN} Data Structure for Curves via Subspaces of Bounded Doubling Dimension},
  author={Conradi, Jacobus and Driemel, Anne and Kolbe, Benedikt},
  journal={Computing in Geometry and Topology},
  volume={3},
  number={2},
  pages={6--1},
  year={2024},
    doi = {10.57717/cgt.v3i2.45}
}

@inproceedings{bruning2022faster,
  title={Faster Approximate Covering of Subcurves Under the {F}r{\'e}chet Distance},
  author={Br{\"u}ning, Frederik and Conradi, Jacobus and Driemel, Anne},
  booktitle={Proc. 30th Annual European Symposium on Algorithms (ESA)},
  pages={28--1},
 ISBN =	{978-3-95977-247-1},
  ISSN =	{1868-8969},
  year={2022},
  doi =		{10.4230/LIPIcs.ESA.2022.28},
}

@inproceedings{colombe21continuous_frechet,
  author    = {Connor Colombe and
               Kyle Fox},
  title     = {Approximating the (Continuous) {F}r{\'{e}}chet Distance},
  booktitle = {Proc. 37th International Symposium on Computational Geometry (SoCG)},
  pages     = {26:1--26:14},
  year      = {2021},
  doi       = {10.4230/LIPIcs.SoCG.2021.26},
}

@inproceedings{vanderhorst24faster_approximation,
  author       = {Thijs van der Horst and
                  Tim Ophelders},
  title        = {Faster {F}r{\'{e}}chet Distance Approximation Through Truncated Smoothing},
  booktitle    = {Proc. 40th International Symposium on Computational Geometry ({SoCG})},
  pages        = {63:1--63:15},
  year         = {2024},
  doi          = {10.4230/LIPICS.SOCG.2024.63},
}

@article{bringmann16approx_discrete_frechet,
  author       = {Karl Bringmann and
                  Wolfgang Mulzer},
  title        = {Approximability of the discrete {F}r{\'{e}}chet distance},
  journal      = {Journal of Computational Geometry},
  volume       = {7},
  number       = {2},
  pages        = {46--76},
  year         = {2016},
  doi          = {10.20382/jocg.v7i2a4},
}

@article{buchin17continuous_frechet,
  author    = {Kevin Buchin and
               Maike Buchin and
               Wouter Meulemans and
               Wolfgang Mulzer},
  title     = {Four Soviets Walk the Dog: Improved Bounds for Computing the {F}r{\'{e}}chet Distance},
  journal   = {Discrete \& Computational Geometry},
  volume    = {58},
  number    = {1},
  pages     = {180--216},
  year      = {2017},
  doi       = {10.1007/s00454-017-9878-7},
}

@inproceedings{cheng25subquadratic_frechet,
  author       = {Siu{-}Wing Cheng and
                  Haoqiang Huang},
  title        = {{F}r{\'{e}}chet Distance in Subquadratic Time},
  booktitle    = {Proc. 2025 Annual {ACM-SIAM} Symposium on Discrete Algorithms (SODA)},
  pages        = {5100--5113},
  year         = {2025},
  doi          = {10.1137/1.9781611978322.173},
}

@inproceedings{bringmann14hardness,
  author    = {Karl Bringmann},
  title     = {Why Walking the Dog Takes Time: {F}r{\'{e}}chet Distance Has No Strongly Subquadratic Algorithms Unless {SETH} Fails},
  booktitle = {Proc. 55th Annual Symposium on Foundations of Computer Science (FOCS)},
  pages     = {661--670},
  year      = {2014},
  doi       = {10.1109/FOCS.2014.76},
}

@inproceedings{buchin19seth_says,
  author    = {Kevin Buchin and
               Tim Ophelders and
               Bettina Speckmann},
  title     = {{SETH} Says: Weak {F}r{\'{e}}chet Distance is Faster, but only if it is Continuous and in One Dimension},
  booktitle = {Proc. 30th Annual {ACM-SIAM} Symposium on Discrete Algorithms (SODA)},
  pages     = {2887--2901},
  year      = {2019},
  doi       = {10.1137/1.9781611975482.179},
}

@article{gudmundsson19long,
  author    = {Joachim Gudmundsson and
               Majid Mirzanezhad and
               Ali Mohades and
               Carola Wenk},
  title     = {Fast {F}r{\'{e}}chet Distance Between Curves with Long Edges},
  journal   = {International Journal of Computational Geometry \& Applications},
  volume    = {29},
  number    = {2},
  pages     = {161--187},
  year      = {2019},
  doi       = {10.1142/S0218195919500043},
}

@article{alt04planar_curves,
  author       = {Helmut Alt and
                  Christian Knauer and
                  Carola Wenk},
  title        = {Comparison of Distance Measures for Planar Curves},
  journal      = {Algorithmica},
  volume       = {38},
  number       = {1},
  pages        = {45--58},
  year         = {2004},
  doi          = {10.1007/S00453-003-1042-5},
}

@article{gudmundsson23approximating_packedness,
  author       = {Joachim Gudmundsson and
                  Yuan Sha and
                  Sampson Wong},
  title        = {Approximating the packedness of polygonal curves},
  journal      = {Computational Geometry},
  volume       = {108},
  pages        = {101920},
  year         = {2023},
  doi          = {10.1016/J.COMGEO.2022.101920},
}

@article{gudmundsson:mapmatching,
  author       = {Joachim Gudmundsson and
                  Martin P. Seybold and
                  Sampson Wong},
  title        = {Map Matching Queries on Realistic Input Graphs Under the {F}r{\'{e}}chet Distance},
  journal      = {{ACM} Transactions on Algorithms},
  volume       = {20},
  number       = {2},
  pages        = {14},
  year         = {2024},
  doi          = {10.1145/3643683},
}

@article{schwarzkopf96range_searching_low_density,
  author       = {Otfried Schwarzkopf and
                  Jules Vleugels},
  title        = {Range Searching in Low-Density Environments},
  journal      = {Information Processing Letters},
  volume       = {60},
  number       = {3},
  pages        = {121--127},
  year         = {1996},
  doi          = {10.1016/S0020-0190(96)00154-8},
}

@inproceedings{wenk06map_matching,
  author       = {Carola Wenk and
                  Randall Salas and
                  Dieter Pfoser},
  title        = {Addressing the Need for Map-Matching Speed: Localizing Global Curve-Matching Algorithms},
  booktitle    = {Proc. 18th International Conference on Scientific and Statistical Database Management (SSDBM)},
  pages        = {379--388},
  year         = {2006},
  doi          = {10.1109/SSDBM.2006.11},
}

@inproceedings{driemel16clustering,
  author       = {Anne Driemel and
                  Amer Krivosija and
                  Christian Sohler},
  title        = {Clustering time series under the {F}r{\'{e}}chet distance},
  booktitle    = {Proc. Twenty-Seventh Annual {ACM-SIAM} Symposium on Discrete Algorithms (SODA)},
  pages        = {766--785},
  year         = {2016},
  doi          = {10.1137/1.9781611974331.CH55},
}

@inproceedings{cheng25subquadratic,
author = {Cheng, Siu-Wing and Huang, Haoqiang and Zhang, Shuo},
title = {Constant Approximation of {F}r\'{e}chet Distance in Strongly Subquadratic Time},
year = {2025},
doi = {10.1145/3717823.3718157},
booktitle = {Proc. 57th Annual ACM Symposium on Theory of Computing (STOC)},
pages = {2329–2340},
}

@inproceedings{vanderhorst23subquadratic_frechet,
  author       = {Thijs van der Horst and
                  Marc J. van Kreveld and
                  Tim Ophelders and
                  Bettina Speckmann},
  title        = {A Subquadratic $n^\varepsilon$-approximation for the Continuous {F}r{\'{e}}chet Distance},
  booktitle    = {Proc. 2023 {ACM-SIAM} Symposium on Discrete Algorithms (SODA)},
  pages        = {1759--1776},
  year         = {2023},
  doi          = {10.1137/1.9781611977554.CH67},
}

\appendix
\end{document}